\documentclass[a4paper]{article}
\usepackage{amsfonts}
\usepackage[english]{babel}
\usepackage{graphicx}
\usepackage[ruled,vlined]{algorithm2e}
\usepackage{multirow}
\usepackage{amsmath,amssymb,amsthm}
\usepackage{xspace}
\usepackage{tikz}
\usepackage{booktabs}
\usetikzlibrary{decorations.pathreplacing,calc,positioning,fit}
\usepackage{subfig}
\usepackage[nolist]{acronym}
\usepackage{fullpage}
\usepackage{url}

\newcommand{\KCG}[2]{\ensuremath{\mathcal{C}_{#1}(#2)}\xspace}
\newcommand{\KKG}[2]{\ensuremath{\mathcal{K}_{#1}(#2)}\xspace}

\newcommand{\period}{\ensuremath{p}\xspace}
\newcommand{\periods}{\ensuremath{P}\xspace}
\newcommand{\event}{\ensuremath{e}\xspace}
\newcommand{\events}{\ensuremath{E}\xspace}
\newcommand{\room}{\ensuremath{r}\xspace}
\newcommand{\rooms}{\ensuremath{R}\xspace}

\newcommand{\conflicts}{\ensuremath{L}\xspace}

\newcommand{\resources}{\ensuremath{\periods\times \rooms}\xspace}

\newcommand{\kemperecolor}{\textsc{Kempe\-Re\-con\-figu\-ration}\xspace}
\newcommand{\vertexelimination}{\textsc{Vertex\-Elimination}\xspace}
\newcommand{\instance}{\ensuremath{\mathcal{I}}\xspace}
\newcommand{\powerset}{\ensuremath{\mathcal{P}}}
\newcommand{\Glist}{\ensuremath{H_G}\xspace}

\newcommand{\kx}{\ensuremath{\kappa}\xspace}
\newcommand{\timetable}{\ensuremath{\tau}\xspace}
\newcommand{\availability}{\ensuremath{\alpha}\xspace}

\newcommand{\roomassignment}{\ensuremath{\rho}\xspace}
\newcommand{\nat}{\ensuremath{\mathbb{N}}\xspace}
\newcommand{\uline}[1]{\overset{#1}{\ \rule[3pt]{1em}{0.5pt}\ }}
\newcommand{\diameter}{\operatorname{diam}}
\newcommand{\ordering}{\ensuremath{\sigma}\xspace}
\newcommand{\orderings}{\ensuremath{S}\xspace}
\newcommand{\subdegeneracy}{\ensuremath{\operatorname{subdeg}}\xspace}
\newcommand{\subdegeneracyU}{\ensuremath{\operatorname{subdeg_{\mathit ub}}}\xspace}
\newcommand{\degeneracy}{\ensuremath{\operatorname{deg}}}
\newcommand{\mindegree}{\ensuremath{\delta}\xspace}
\newcommand{\vdegree}{\ensuremath{d}\xspace}

\newcommand{\eq}{\ensuremath{\mathrel{\sim}}\xspace}
\newcommand{\ereq}{\ensuremath{\mathrel{\sim_{E}}}\xspace}
\newcommand{\keq}{\ensuremath{\mathrel{\sim_{K}}}\xspace}
\newcommand{\feq}{\ensuremath{\mathrel{\sim_{f}}}\xspace}
\newcommand{\fneq}{\ensuremath{\mathrel{\not\sim_{f}}}\xspace}
\newcommand{\argmin}{\operatornamewithlimits{argmin}}

\newcommand{\erlangen}{\texttt{erlangen}\xspace}

\newcommand{\pspace}{\ensuremath{\mathrm{PSPACE}}\xspace}

\theoremstyle{definition}
\newtheorem{theorem}{Theorem}

\newtheorem{corollary}{Corollary}
\newtheorem{definition}{Definition}
\newtheorem{proposition}{Proposition}

\newtheorem{remark}{Remark}
\makeatletter
\newdimen\@@@tmpa
\newdimen\@@@tmpb
\newdimen\@@@tmpc
\newdimen\@@@tmpd
\def\clap#1#2#3{
       \setbox0=\hbox{#1}\setbox1=\hbox{#2}%
       \@@@tmpa\wd0\@@@tmpb\wd1\advance\@@@tmpa-\@@@tmpb%
       \ifdim\@@@tmpa>0pt\@@@tmpb\wd0%
       \else\@@@tmpb\wd1\fi\@@@tmpc\ht0\@@@tmpd\ht1%
       \advance\@@@tmpc\dp0\advance\@@@tmpd\dp1%
       \advance\@@@tmpc-\@@@tmpd\divide\@@@tmpc2%
       \ifdim\@@@tmpc>0pt\@@@tmpd\ht0\advance\@@@tmpd\dp0\@@@tmpc-\dp0%
       \else\@@@tmpd\ht1\advance\@@@tmpd\dp1\advance\@@@tmpc-\dp0\fi%
       \ifx#3\empty\else\advance\@@@tmpc#3\fi%
       \leavevmode\raise\@@@tmpc\hbox to \@@@tmpb{\rlap{\hbox to \@@@tmpb{\hss%
            \vbox to \@@@tmpd{\vss\box0\vss}\hss}}%
            \hss\vbox to \@@@tmpd{\vss\box1\vss}\hss}%
       }

\makeatother

\begin{document}

\pagestyle{empty}

\title{On the Connectedness of Clash-free Timetables\footnote{This is the extended version of~\cite{MW:14} presented at PATAT 2014.}
}

\author{Moritz M\"uhlenthaler\thanks{Research funded in parts by the School of
Engineering of the University of Erlangen-Nuremberg.} \and Rolf Wanka}

\maketitle

\begin{acronym}[]
\acro{ACO}{Ant Colony Optimization}
\acro{ACP}{Algorithm Configuration Problem}
\acro{ASP}{Answer Set Programming}
\acro{AMOSA}{Archived Multi-Objective Simulated Annealing}
\acro{BACP}{Balanced Academic Curriculum Problem}
\acro{BFS}{breadth-first search}
\acro{BSSP}{Basic Student Sectioning Problem}
\acro{CB-CTT}{Curriculum-based Course Timetabling}
\acro{CNF}{conjunctive normal form}
\acro{DPLL}{Davis–Putnam–Logemann–Loveland}
\acro{EA}{Evolutionary Algorithm}
\acro{EIP}{Event Insertion Problem}
\acro{ETP}{Examination Timetabling Problem}
\acro{ETT}{Examination Timetabling}
\acro{FAU}{Friedrich-Alexander-Universit{\"a}t}
\acro{GA}{Genetic Algorithms}
\acro{GBACP}{Generalized Balanced Academic Curriculum Problem}
\acro{GD}{Great Deluge}
\acro{GGA}{Grouping Genetic Algorithm}
\acro{GLBOP}{Generalized Lexicographic Bottleneck Optimization Problem}
\acro{HC}{Hill Climbing}
\acro{HGHH}{Hybrid Graph-based Hyper Heuristic}
\acro{ILP}{Integer Linear Programming}
\acro{ILS}{Iterated Local Search}
\acro{IP}{Integer Programming}
\acro{ITC2007}{International Timetabling Competition 2007}
\acro{JFI-CB-CTT}{Jain's Fairness Index Curriculum-based Course Timetabling}
\acro{KIH}{Kempe Insertion Heuristic}
\acro{KX}{Kempe-exchange Neighborhood}
\acro{LBAP}{Lexicographic Bottleneck Assignment Problem}
\acro{LBFS}{lexicographic breadth-first search}
\acro{LBOP}{Lexicographic Bottleneck Optimization Problem}
\acro{LCD}{Largest Color Degree}
\acro{LD}{Largest Degree}
\acro{LSAP}{Linear Sum Assignment Problem}
\acro{LSD}{Least Saturation Degree}
\acro{LS}{Local Search}
\acro{SLS}{Stochastic Local Search}
\acro{LVOP}{Lexicographic Vector Optimization Problem}
\acro{LWD}{Largest Weighted Degree}
\acro{ME}{Move Event Neighborhood}
\acro{MMF-CB-CTT}{Max-min Fair Curriculum-based Course Timetabling}
\acro{NOLH}{Nearly Orthogonal Latin Hypercubes}
\acro{PE-CTT}{Post-enrollment Course Timetabling}
\acro{PKC}{Pair-wise Kempe-Chain}
\acro{PSO}{Particle Swarm Optimization}
\acro{RO}{Random Ordering}
\acro{SA}{Simulated Annealing}
\acro{SAT}{Boolean satisfiability}
\acro{SE}{Swap Event Neighborhood}
\acro{SOP}{Min-sum Optimization Problem}
\acro{SO}{Saturation Ordering}
\acro{SSP}{Student Sectioning Problem}
\acro{STP}{School Timetabling Problem}
\acro{TEICH}{Tabu-seach Event Insertion Construction Heuristic}
\acro{TF-CB-CTT}{TechFak Curriculum-based Course Timetabling}
\acro{TS}{Tabu Search}
\acro{uar}{uniformly at random}
\acro{UCTP-OPT}{University Course Timetabling Optimization Problem}
\acro{UCTP}{University Course Timetabling Problem}
\acro{UTP}{University Timetabling Problem}
\acro{VNS}{Variable Neighborhood Search}
\acro{WEO}{weak elimination ordering}
\acro{WMW}{Wilcoxon-Mann-Whitney}
\end{acronym}

\begin{abstract}

We investigate the connectedness of clash-free timetables with respect to the
Kempe-exchange operation. This investigation is related to the connectedness of
the search space of timetabling problem instances, which is a desirable
property, for example for two-step algorithms using the Kempe-exchange during
the optimization step. The theoretical framework for our investigations is
based on the study of reconfiguration graphs, which model the search space of
timetabling problems. We contribute to this framework by including timeslot
availability requirements in the analysis and we derive improved
conditions for the connectedness of clash-free timetables in this setting. 
We apply the theoretical
insights to establish the connectedness of clash-free timetables for a number
of benchmark instances.

\end{abstract}


\section{Introduction}
\label{sec:Intro}


Timetabling problems in the context of a university ask for an assignment of
events (e.g., courses or exams) to rooms and timeslots such that no two
conflicting events are scheduled simultaneously.  By a straightforward
reduction from vertex coloring it is immediate that such timetabling problems
are NP-hard, which motivates the use of (meta-)heuristics in order to solve
timetabling problems in practice; see for example~\cite{Schaerf:99} for an
overview of different problem models and solution approaches.
According to the classification of heuristic optimization algorithms
for timetabling problems in~\cite{Lewis:phd}, many approaches in the literature
fall in the category of \emph{two-step optimization algorithms}. The general
procedure is the following: In the first step, the underlying search problem is
solved and the resulting feasible solution is used as a starting point for the
second step, during which the optimization is performed. In the second step
only feasible solutions are considered. A recent example of a state-of-the-art
two-step approach is~\cite{LH:10}, numerous other examples can be found
in~\cite{Lewis:phd}. During the optimization step, feasible timetables are
modified using Kempe-exchanges or similar operations that preserve their
feasibility. It is natural to ask whether any feasible timetable, in particular
an optimal one, can be reached from an initial feasible timetable. We give a
partial answer to this question by investigating conditions that establish the
connectedness of the search space of \emph{clash-free} timetables. 

A timetable is clash-free, if no two conflicting events are scheduled
simultaneously. We model the structure of the search space of
clash-free timetables in terms of reconfiguration graphs. Such graphs
arise in the context of reconfiguration problems: Given an
instance \instance of a combinatorial search problem, the corresponding
\emph{reconfiguration problem} asks whether one feasible solution to \instance
can be transformed into another feasible solution in a step-by-step manner by
making local changes, such that each intermediate solution is also feasible. A
\emph{reconfiguration graph} has as nodes the feasible solutions of the
underlying combinatorial problem and two such solutions are adjacent whenever
there is a local change that transforms one into the other.
Reconfiguration variants of classical combinatorial problems and their
reconfiguration graphs have been studied in the literature ~\cite[see
e.g.]{Bonsma:12,GKMP:09,IDHP:11,KMM:11,Kaminski:12}. 
The heart of the matter
of timetabling problems in the academic context (in contrast to the ``high school timetabling'' model, see~\cite[Section 2]{deWerra:85}, \cite{Pillay:14}) is the vertex coloring problem:
A clash-free timetable corresponds to a proper coloring of the event conflict
graph, see e.g.~\cite{deWerra:85}. \cite{CHJ:08} have shown that determining
the connectedness of any two proper colorings of a graph is \pspace-complete for
four or more colors and tractable otherwise, in a setting where an admissible
local change alters the color of a single vertex. 
A related line of research
deals with the question whether two given proper colorings are connected, see
e.g.,~\cite{BC:09,Wrochna:15}.
\cite{VM:81} 
give sufficient (but not necessary) conditions for the connectedness of any two
vertex colorings with respect to the Kempe-exchange operation.   
The Kempe-exchange is a generalization of the local change mentioned above. It
is a popular operation used by algorithms for timetabling problems for
exploring the search space, including many of the two-step algorithms in the
references above. Therefore, the results of~\cite{VM:81} are applicable in the
timetabling context, see Corollary~\ref{cor:kkg_connectedness}.   Basically,
the clash-free timetables are connected if the number of timeslots is
sufficiently large compared to the degeneracy of the graph of event conflicts.
Fortunately, this condition can be checked efficiently.

Clash-freeness is typically not the only requirement for
a timetable to be feasible. In many problem formulations in research and
practice~\cite[e.g.,]{deWerra:96,Carter:00,SH:07,Bonutti:12}, certain
timeslots or rooms may be unavailable/unsuitable for particular events and it
is required that each event is placed strictly in the available rooms and
timeslots. We employ a standard reduction from list to vertex coloring in order
to include timeslot availability requirements in the reconfiguration model. We  
show that this approach leads to a faithful representation of the search space
with respect to its connectedness and its diameter. It turns out that due to
the nature of the reduction the condition in
Corollary~\ref{cor:kkg_connectedness} is too strict to be useful for certifying
the connectedness of the clash-free timetables that satisfy the timeslot
availability requirements.  However, we extend the techniques
from~\cite{VM:81} to derive improved conditions in this setting.
For this purpose, we introduce the
\emph{subdegeneracy} of a graph, which generalizes the notion of degeneracy by
ignoring the potential contribution to the degeneracy of a given subgraph. We
show that the clash-free timetables that satisfy the timeslot availability
requirements are connected with respect to the Kempe-exchange if there are
sufficiently many timeslots compared to the subdegeneracy of the conflict graph
and a suitably chosen subgraph.  In contrast to the degeneracy, which can be
computed in linear time, the computational complexity of determining the
subdegeneracy is an open problem and we propose a heuristic solution approach.
We further provide data on the connectedness of the clash-free timetables for a
number of benchmarking instance sets, including artificial and real-world
instance, with and without taking timeslot availability requirements into
account. 
%

The remainder of this work is organized as follows: In
Section~\ref{sec:background} we provide the basic formalisms required for our
analysis of the connectedness of clash-free timetables presented in
Section~\ref{sec:connectedness}. In Section~\ref{sec:results} we investigate
the connectedness of the clash-free timetables for number of standard
benchmarking instance sets.  


\section{Background}
\label{sec:background}

\subsection{The \acl{UTP}}
\label{sec:background:uctp}

The \acf{UTP} formalizes in terms of a search problem the task of creating a
course or examination schedule at a university.

\begin{definition}[\acf{UTP}]
\label{def:uctp-search}
\mbox{}\newline
\emph{INSTANCE:}
\begin{itemize}
	\item a set of events $\events = \{ \event_1,\dotsc,\event_n \}$ 
	\item a set of rooms $\rooms = \{\room_1,\dotsc,\room_\ell\}$
	\item a set of timeslots $\periods = \{\period_1,\dotsc,\period_k\}$
	\item a graph $G = (\events, \conflicts)$ with nodes \events 
and edges $\conflicts \subseteq \{ \{u, v\} \mid u, v \in \events \}$
\end{itemize}
The graph $G$ is referred to as the \emph{conflict graph}\index{conflict graph}. 
Two events are called \emph{conflicting}\index{conflict} if they are adjacent in $G$.
An element of the set \resources is referred to as \emph{resource}\index{resource}.  
A \emph{timetable} \index{timetable} $\timetable$ is an assignment $\tau :
\events \rightarrow \resources$. Two events $\event, \event'$ are
\emph{overlapping}\index{overlap}, if $\event \neq \event'$ and $\tau(\event) =
\tau(\event')$. A timetable is called \emph{overlap-free}\index{overlap-free
timetable} if no two events overlap. Two events $\event, \event'$ are
\emph{clashing} \index{clash} in $\timetable$, if they are conflicting and they
are assigned to the same timeslot. A timetable is \emph{feasible}, if it is
clash-free and overlap-free.
\newline
\emph{TASK:}\;Find a feasible timetable.
\end{definition}

It is usually assumed that all timeslots have the same length and that 
each event fits in a single timeslot. 
The \ac{UTP} as defined above is equivalent to the problem given in~\cite[Section
3.4]{deWerra:85} and generalizes many of the more refined problem formulations
in the literature~\cite[see e.g.]{Bonutti:12,ITC2007:CB-CTT,ITC2007:ETT}. The
clash-freeness requirement and its relation to the vertex coloring problem is
the heart of the matter of timetabling problems in the academic
context~\cite[see e.g.]{deWerra:85,Schaerf:99}. Other kinds of requirements
such as \emph{availability requirements} and \emph{precedence
requirements} often occur in practice~\cite[e.g.]{Carter:00,SH:07} and in
the benchmarking problem models~\cite[e.g.]{Bonutti:12,ITC2007:CB-CTT,ITC2007:ETT}.
Later, we will consider the \ac{UTP} above with additional timeslot availability
requirements. These requirements mandate that only specific timeslots can be
assigned to an event. We formalize timeslot availability requirements in terms of
an availability function \availability, which determines for each event the set
of available timeslots 
\[
\availability: \events \rightarrow \powerset(\periods)\enspace.
\]

An important subproblem of the \ac{UTP} is the \emph{room assignment problem}.
Given a timeslot $\period \in \periods$, then events $\events' \subseteq \events$
admit a room assignment, if there is an assignment $\roomassignment:
\events' \rightarrow \rooms$ such that $(\period, \roomassignment(\event))$ is
available for each $\event \in \events'$.

\subsection{Vertex Coloring}
\label{sec:background:vcol}

A graph $G = (V(G), E(G))$, for short $G = (V, E)$, consists of a set of
\emph{vertices} $V$ and a set of \emph{edges} $E \subseteq \{ \{u, v\} \mid u,
v \in V\}$.  Unless stated otherwise, we assume that graphs are loopless and
finite. We denote by $u\uline{}v$ that the vertices $u$ and $v$ are adjacent,
i.e., $\{u,v\} \in E$.  The graph $G[U]$ denotes the subgraph of $G$ induced by
the vertices $U \subseteq V(G)$.  
A \emph{\mbox{(vertex-)}$k$-coloring} of a graph
$G$ is a mapping $c: V \rightarrow \{1,\dotsc,k\}$ that assigns one of the
colors $\{1,\dotsc,k\}$ to each vertex of $G$. A coloring is called
\emph{proper}\index{proper coloring}, if no two adjacent nodes have the same
color. Unless stated otherwise, we will use the term \emph{coloring} as a
shorthand for \emph{proper coloring}. The \emph{vertex $k$-coloring problem} asks,
whether a graph admits a $k$-coloring. A $k$-coloring $c$ of $G$ partitions the
vertices of $G$ into $k$ sets of independent (mutually non-adjacent) vertices called \emph{color classes}.  A color class
$a \in \{1,\ldots,k\}$ contains all vertices of color $a$. We denote by $G(a,
b)$ the bipartite subgraph induced by the color classes $a$ and $b$. A
connected component in $G(a, b)$ is referred to as \emph{Kempe-component}.


Given a set $L(v)$ (called \emph{list}) of available colors for each $v \in V$,
a \emph{list coloring} $c: V \rightarrow \bigcup_{v \in V} L(v)$ of $G$ is a 
coloring of $G$ such that $c(v) \in L(v)$ for each $v \in V$.  Vertex coloring
is a special case of list coloring, where all colors are available for each
node. By using a standard technique~\cite[Proposition
3.2]{deWerra:85} list coloring can be reduced to vertex coloring: Let the
colors be labeled $1,\ldots,k$, where $k = \left\vert \bigcup_{v \in V}
L(v)\right\vert$. Now, let the graph $\Glist$ be a copy of $G$ to which we add a
clique $C$ on $k$ (new) nodes $v_1,\ldots,v_k$. For each $v \in V(G)$, we add an edge
$v\uline{}v_i$ to $\Glist$, whenever $i \notin L(v)$. Clearly, $\Glist$ admits a
$k$-coloring if and only if $G$ admits a list coloring. The problem of deciding
if a given \ac{UTP} instance admits a clash-free timetable that satisfies
timeslot availability requirements is equivalent to deciding if the conflict
graph admits a list coloring, where the $L(\event) = \alpha(\event)$ for each
event \event.

\subsection{The Vertex Coloring Reconfiguration Problem}
\label{sec:background:reconfiguration}

Reconfiguration problems formalize the question, if a solution to a problem
instance can be transformed into another solution in a step-by-step manner by
some reconfiguration operation, such that each intermediate solution is
feasible~\cite{IDHP:11}. To show that a search space is connected we
need to check whether \emph{any two} solutions are connected. In the context
of the vertex coloring problem this question has been investigated for example
in~\cite{Mohar:07,CHJ:08,BC:09,BJLPP:14,FJP:15}. As a reconfiguration operation,
elementary recolorings and Kempe-exchanges have been considered
in the literature. Given a coloring $c$ of a graph
$G$, an \emph{elementary recoloring} changes the color of a single vertex $u$
of $G$ to a color that does not occur in the neighborhood of $u$. Two
$k$-colorings $c_1$ and $c_2$ of $G$ are adjacent, $c_1 \ereq c_2$, if there is
an elementary recoloring that transforms $c_1$ into $c_2$. The Kempe-exchange
is a generalization of the elementary recoloring operation. Given two colors
$a$ and $b$, a Kempe-exchange switches the colors of a Kempe-component, i.e., a
connected component in $G(a, b)$. The result of this operation is a new
coloring, such that, within the Kempe-component, each vertex of the of color
$a$ is assigned to color $b$ and vice versa. An elementary recoloring that
changes the color of a vertex $u$ from $a$ to $b$ is a Kempe-exchange on the
Kempe-component containing $u$ in $G(a,b)$, which is an isolated vertex. 
Two colorings $c_1$ and $c_2$ of
$G$ are adjacent with respect to the Kempe-exchange, denoted by $c_1 \keq c_2$,
if there is a Kempe-exchange that transforms $c_1$ into $c_2$. Each of the two
adjacency relations ${\ereq}$ and ${\keq}$ gives rise to a graph structure on
the set of $k$-colorings of $G$.

\begin{figure}
	\begin{center}
	\begin{tikzpicture}[node distance=10em,vertex/.style={shape=circle,draw,scale=0.5,fill=black},container/.style={shape=circle,draw,inner sep=2pt},rgedge/.style={thick}]
		\node[container] (G1) {
			\begin{tikzpicture}[node distance=5em]
				\node[vertex,label=left:{$1$}] (v1) {};
				\node[vertex,below of=v1,label=left:{$2$}] (v2) {};
				\draw (v1) -- (v2);
			\end{tikzpicture}
		};
		\node[container,right of=G1] (G2) {
			\begin{tikzpicture}[node distance=5em]
				\node[vertex,label=left:{$1$}] (v1) {};
				\node[vertex,below of=v1,label=left:{$3$}] (v2) {};
				\draw (v1) -- (v2);
			\end{tikzpicture}
		};
		\node[container,right of=G2] (G3) {
			\begin{tikzpicture}[node distance=5em]
				\node[vertex,label=left:{$2$}] (v1) {};
				\node[vertex,below of=v1,label=left:{$3$}] (v2) {};
				\draw (v1) -- (v2);
			\end{tikzpicture}
		};
		\node[container,below of=G1] (G4) {
			\begin{tikzpicture}[node distance=5em]
				\node[vertex,label=left:{$3$}] (v1) {};
				\node[vertex,below of=v1,label=left:{$2$}] (v2) {};
				\draw (v1) -- (v2);
			\end{tikzpicture}
		};
		\node[container,below of=G2] (G5) {
			\begin{tikzpicture}[node distance=5em]
				\node[vertex,label=left:{$3$}] (v1) {};
				\node[vertex,below of=v1,label=left:{$1$}] (v2) {};
				\draw (v1) -- (v2);
			\end{tikzpicture}
		};
		\node[container,below of=G3] (G6) {
			\begin{tikzpicture}[node distance=5em]
				\node[vertex,label=left:{$2$}] (v1) {};
				\node[vertex,below of=v1,label=left:{$1$}] (v2) {};
				\draw (v1) -- (v2);
			\end{tikzpicture}
		};

		\draw[rgedge] (G1) -- (G2) -- (G3);
		\draw[rgedge] (G4) -- (G5) -- (G6);
		\draw[rgedge] (G1) -- (G4);
		\draw[rgedge] (G3) -- (G6);
		\draw[rgedge,dashed] (G2) -- (G5);
		\draw[rgedge,dashed] (G1) -- (G6);
		\draw[rgedge,dashed] (G3) -- (G4);
	\end{tikzpicture}
	\end{center}
	\caption{The Kempe-$3$-coloring graph $\KKG{3}{K_2}$ of the graph $K_2$. The subgraph
	  induced by the solid edges corresponds to $\KCG{3}{K_2}$. Each dashed
	  edge represents a Kempe-exchange that cannot be realized by a single
	  elementary recoloring.\label{fig:background:kkg}}
\end{figure}
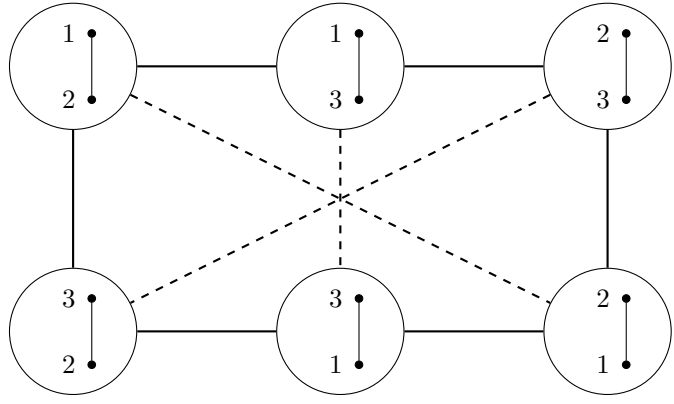

\begin{definition} [(Kempe-)$k$-coloring graph]
  \label{def:kempe_reconfiguration}
  For a graph $G=(V, E)$ and $k \in \nat$ let
  \begin{align*}
	  \mathcal{V}   &:= \{ c: V \rightarrow \{1, \dotsc, k\}  \mid c \text{ is a $k$-coloring of } G \} \\
			  \mathcal{E}_E &:= \{ \{c_1,c_2\} \mid c_1, c_2 \in \mathcal{V} \text{ and } c_1 \ereq c_2 \} \\
				  \mathcal{E}_K &:= \{ \{c_1,c_2\} \mid c_1, c_2 \in \mathcal{V} \text{ and } c_1 \keq c_2 \}\enspace.
  \end{align*}
  Then the $k$-coloring graph is the graph $\KCG{k}{G} = (\mathcal{V},
  \mathcal{E}_E)$. The Kempe-$k$-coloring graph is the graph $\KKG{k}{G} =
  (\mathcal{V}, \mathcal{E}_K)$.
\end{definition}

\begin{algorithm}
  \caption{\kemperecolor}
  \label{alg:kemperecolor}
  \DontPrintSemicolon

  \SetKwInOut{Input}{input}
  \SetKwInOut{InOut}{in/out}
  \SetKwInOut{Output}{output}
  \SetKwInOut{Data}{data}

  \SetKwData{GG}{$G$}
  \SetKwData{HH}{$H$}
  \SetKwData{CC}{$c_1$}
  \SetKwData{DD}{$c_2$}
  \SetKwData{KK}{$K$}
  \SetKwData{COL}{$c$}

  \Input{graph \GG, labeling $v_1,\ldots,v_n$ of the vertices, $k$-colorings \CC, \DD of \GG}
  \Output{list of Kempe-exchanges transforming \CC into \DD}
  \Data{array \COL of length $n$ storing the current color of each vertex, list \KK of Kempe-exchanges}
  \BlankLine
  \KK $\longleftarrow$ empty list\;
  \BlankLine
  \For{$i \longleftarrow 1$ \KwTo $n$}
  {
	$\HH \longleftarrow G[v_1,\ldots,v_i]$\;
	\For{$i \longleftarrow 1$ \KwTo $n$}{$\COL[i] \longleftarrow \CC(v_i)$\;}
	\tcc{Kempe-exchange $\kx = (a, b, u)$, where $a$,$b$ are colors and $u \in V(\HH)$}
	\For{$\kx = (a, b, u) \in \KK$}
	{
		without loss of generality let $a \neq \COL[i]$\;

		\lnl{alg:kemperecolor:one}\If{$\COL[i] = b$ and $v_i$ has exactly one neighbor of color $a$ in \HH}
		{
		  \tcc{Note that the color of $v_i$ will be changed by \kx in \HH}
		}
		\lnl{alg:kemperecolor:two}\If{$\COL[i] = b$ and $v_i$ has at least two neighbors of color $a$ in \HH}
		{
			choose color $b' \neq b$, which is not used by any neighbor of $v_i$ in \HH\;
			insert Kempe-exchange $k = (b, b', v_i)$ right before $\kx$ in \KK  and apply $k$ to \COL\; 
		}

		apply Kempe-exchange \kx to \COL\;
  }
	\lnl{alg:kemperecolor:final}append Kempe-exchange $(c_2(v_i), \COL[i], v_i)$ to \KK\; 
  }

  \Return \KK\;
\end{algorithm}

Figure~\ref{fig:background:kkg} shows as a toy example $\KCG{3}{K_2}$ and
$\KKG{3}{K_2}$, where $K_2$ is the complete two-vertex graph.
Clearly, for any graph $G$ and $k \geq 1$, $\KCG{k}{G} \subseteq \KKG{k}{G}$.
The diameter and the connectedness of (Kempe-)$k$-coloring graphs has been
investigated for example in~\cite{Mohar:07,BJLPP:14,FJP:15}. To the best of
our knowledge, in general graphs and for $k \geq 4$, the complexity of deciding
the connectedness of any two $k$-colorings of a graph is still open. However,
in~\cite{VM:81}, a sufficient (but not necessary) condition for the
connectedness of any two $k$-colorings is given, which relates the
connectedness of the colorings to the degeneracy of the graph to be colored.
A graph $G$ is called \emph{$k$-degenerate}, if its vertices can be linearly ordered
such that each vertex has at most $k$
neighbors preceding it. The smallest $k$ for which $G$ admits such an ordering
is the \emph{degeneracy} $\degeneracy(G)$, which is sometimes also called \emph{width} of $G$. A witness vertex
ordering of $\degeneracy(G)$ can be found in linear time by repeatedly removing
vertices of minimal degree~\cite{Matula:68,SW:68,BZ:03}. Equivalently,
$\degeneracy(G)$ is the largest minimum degree of any subgraph of $G$.  Let
$\orderings(G)$ be the set of orderings of the vertices of $G$ and let
$\operatorname{pred}(v, \ordering)$ denote the number of neighbors of the
vertex $v \in V(G)$ that precede $v$ in the ordering $\ordering \in
\orderings(G)$. In formal terms, the two characterizations of $\degeneracy(G)$
can be stated as follows:
\begin{equation}
\label{eq:degeneracy}
	\degeneracy(G) := \max_{H \subseteq G} \min_{v \in V(H)} \{d_{H}(v)\} = \min_{\ordering \in \orderings(G)} \max_{v \in V(G)} \operatorname{pred}(v, \ordering)\enspace,    
\end{equation}
where $d_H(v)$ denotes the degree of $v$ in $H$.  The degeneracy of a graph is
an upper bound on its chromatic number. In~\cite{VM:81}, degeneracy has been
used in order to establish the connectedness of Kempe-$k$-coloring graphs as follows:

\begin{theorem} [{\cite[Proposition 2.1]{VM:81}}]
\label{thm:kkg_connectedness}
For any graph $G$, the Kempe-$k$-coloring graph $\KKG{k}{G}$ is connected if $k
> \degeneracy(G)$. 
\end{theorem}

The proofs given in~\cite{VM:81} and~\cite{Mohar:07} are essentially an analysis of the
algorithm \kemperecolor shown in Algorithm~\ref{alg:kemperecolor}. This
algorithm transforms a source coloring $c_1$ into a destination coloring $c_2$
by a sequence of Kempe-exchanges, provided that a sufficient number of colors
is available. The vertices are processed one-by-one according to the given
labeling. The idea behind the algorithm is to prevent that changing the color
of the current vertex interferes with colors of the previously processed
vertices. 

\section{The Connectedness of Clash-free Timetables}
\label{sec:connectedness}

We investigate the connectedness of the search space of clash-free timetables
with respect to the Kempe-exchange operation. In the following, let $G$ be the
conflict graph of a \ac{UTP} instance \instance with timeslots
$\{1,\ldots,\period\}$.
We consider timetables that differ only with respect to the
room assignment to be equivalent. Therefore, each $p$-coloring of $G$
corresponds to an equivalence class of clash-free timetables and the
adjacency relation $\keq$ on the $p$-colorings of $G$ induces an adjacency
relation on the equivalence classes of clash-free timetables. As a consequence, 
$\KKG{p}{G}$ represents the search space of clash-free timetables of clash-free
timetables connected by Kempe-exchanges.  If $\KKG{p}{G}$ is connected, then a
two-step algorithm using Kempe-exchanges for search space exploration can reach
an optimal solution from any starting point.  Otherwise, the algorithm may fail
to find an optimal solution due to the structure of the search space.

A sufficient condition establishing the connectedness of clash-free timetables result
directly from Theorem~\ref{thm:kkg_connectedness}:
\begin{corollary}
\label{cor:kkg_connectedness}
The search space of clash-free timetables is connected if $p > \degeneracy(G)$.
\end{corollary}

In most applications however, clash-freeness is not the only requirement a
timetable needs to satisfy. In addition, timeslot availability requirements,
room availability requirements, and overlap-freeness requirements may restrict
the set of feasible timetables, and, as a consequence, limit the search space
to a certain subgraph of $\KKG{p}{G}$. 
In particular, for the additional requirements above, the search space is
restricted to the following nodes of $\KKG{\period}{G}$: 
\begin{enumerate}
	\item timeslot availability requirements:\label{itm:period_availability}
		\[
			C_{\pi} = \{ c \in V(\KKG{\period}{G}) \mid \forall v \in V(G): c(v) \text{ is available for event $v$} \}
		\]
	\item overlap-freeness and room availability requirements:\label{itm:overlap_freeness}
		\[
			C_{\roomassignment} = \{ c \in V(\KKG{\period}{G}) \mid \forall i \in \periods: \text{ color class $i$ admits a room assignment} \}
		\]
\end{enumerate}

Regarding overlap-freeness and room availability
requirements, to the best of our knowledge, the properties of the corresponding
reconfiguration graphs have not been studied so far. The \emph{bounded vertex
$k$-coloring problem} with bound $b \in \nat$ is the problem of coloring a
graph with $k$ colors such that the size of each color class is at most $b$. The
bounded vertex coloring problem has been studied for example
in~\cite{Lucarelli:09} as well as ~\cite{BC:96} in the
setting of unit-time task scheduling on multiple processors, and
in~\cite{deWerra:97} in the timetabling context. If overlap-freeness is
required and no particular room availability requirements are present, then the
graph $\KKG{\periods}{G}[C_\roomassignment]$ is the reconfiguration graph of a
bounded vertex coloring instance. The reconfiguration variant of the bounded
vertex coloring problem seems to be an interesting problem which deserves
further investigation. The situation gets more involved if room availability
requirements are present: Checking if the $k$ events in a color class admit a
room assignment is equivalent to checking if a suitably chosen bipartite graph
admits a matching of cardinality $k$.

We now investigate conditions that certify the connectedness of
$\KKG{p}{G}[C_\pi]$. Using the standard reduction from list coloring to vertex
coloring described in Section~\ref{sec:background:vcol}, we obtain a graph
\Glist that contains the original conflict graph $G$ and a clique on $p$
additional vertices $v_1,\ldots,v_p$, which is used for representing the
available timeslots for each event. Our goal is to show that the
reconfiguration graph of the $p$-colorings of \Glist is a suitable
representation of the search space of clash-free timetables that satisfy given
timeslot availability requirements. First, we show that $\KKG{p}{\Glist}$ is
connected if and only if $\KKG{p}{G}[C_\pi]$ is connected. Please note that,
due to the nature of the reduction, there are Kempe-exchanges on $p$-colorings
of \Glist for which there is no corresponding Kempe-exchange on $G$; just
consider Kempe-exchanges that involve the nodes $v_1,\ldots,v_p$. We give
further evidence that $\KKG{p}{\Glist}$ is a suitable representation of the
search space $\KKG{p}{G}[C_\pi]$, by showing that their diameters differ only
by a factor linear in $|V(G)|$.

\subsection{Search space representation in terms of $\KKG{p}{\Glist}$}

We first show that $\KKG{p}{\Glist}$ is connected precisely when
$\KKG{p}{G}[C_\pi]$ is connected. For this purpose, we construct from
$\KKG{p}{G}[C_\pi]$ an auxiliary graph $K$.  The vertices of $K$ are the
vertices of $\KKG{p}{G}[C_\pi]$. Any two vertices $u, v \in V(K)$ (i.e.,
colorings of $G$) are adjacent if there is an $u \uline{} v$ edge in
$\KKG{p}{G}$ or if there are two colors $i$ and $j$ such that $u$ can be
transformed into $v$ by swapping the colors in all except a single connected
component of $G(i,j)$.  For technical reasons that involve the construction of
a graph homomorphism to $K$ we add a self-loop to each node of $K$.  From the
construction of $K$ it follows easily that $K$ is connected if and only if
$\KKG{p}{G}$ is connected:

\begin{proposition}
	$K$ is connected if and only if $\KKG{p}{G}$ is connected.
\end{proposition}
\begin{proof}
The edges which occur in $K$ but not in $G$ are merely shortcuts for several
individual Kempe-exchanges performed on $G$. Therefore, $\KKG{p}{G}[C_\pi]$ is
connected if and only if $K$ is connected.    
\end{proof}

Figure~\ref{fig:homomorphism_example} shows the various graphs under
consideration and how they are related for a small list-coloring instance
consisting of a graph $G=(\{u, v\}, \{u\uline{}v\})$ and color lists
$L(u)=\{1\}$ and $L(v)=\{2\}$.  The nodes $v_1$ and $v_2$ of $\Glist$ were
added by the reduction from list coloring to vertex coloring. 

\begin{figure}
	\begin{center}
	\begin{tikzpicture}[node distance=12em,vertex/.style={shape=circle,draw,scale=0.5,fill=black},container/.style={shape=circle,thick,draw,inner sep=2pt},rgedge/.style={thick},every loop/.style={}]
		\node[] (G) at (0, 0) [label=above:$G$] {
			\begin{tikzpicture}
				\node[vertex,label=above:{$u$}] (v1) {};
				\node[vertex,right=2em of v1,label=above:{$v$}] (v2) {};
				\draw (v1) -- (v2);
			\end{tikzpicture}
		};
		\node[left=1em of G,align=center] () {$L(u) = \{1\}$\\$L(v) = \{2\}$};
		
		\node[below of=G] (G2) [label=below:$\Glist$] {
			\begin{tikzpicture}
					\node[vertex,label=above:{$u$}] (v1) {};
					\node[vertex,right=2em of v1,label=above:{$v$}] (v2) {};
					\node[vertex,below=2em of v1,label=below:{$v_1$}] (v3) {};
					\node[vertex,below=2em of v2,label=below:{$v_2$}] (v4) {};
					\draw (v1) -- (v2);
					\draw (v3) -- (v4);
					\draw (v1) -- (v4);
					\draw (v2) -- (v3);
			\end{tikzpicture}
		};
		\node[] (K) at (6.5, 0) {
			\begin{tikzpicture}
				\node[container,label=above:$K$] (N) {
				\begin{tikzpicture}
					\node[vertex,label=above:{$1$}] (v1) {};
					\node[vertex,right=2em of v1,label=above:{$2$}] (v2) {};
					\draw (v1) -- (v2);
				\end{tikzpicture}
				} edge [thick,in=30,out=55,loop,label=above:{}] ();
			\end{tikzpicture}
		}; 

		\node[below of=K] (KGP) [label=below:$\KKG{2}{\Glist}$] { 
			\begin{tikzpicture}
				\node[container] (N1) {
					\begin{tikzpicture}
						\node[vertex,label=above:{$1$}] (v1) {};
						\node[vertex,right=2em of v1,label=above:{$2$}] (v2) {};
						\node[vertex,below=2em of v1,label=below:{$1$}] (v3) {};
						\node[vertex,below=2em of v2,label=below:{$2$}] (v4) {};
						\draw (v1) -- (v2);
						\draw (v3) -- (v4);
						\draw (v1) -- (v4);
						\draw (v2) -- (v3);
					\end{tikzpicture}
				};
				\node[right=2em of N1,container] (N2) {
					\begin{tikzpicture}
						\node[vertex,label=above:{$2$}] (v1) {};
						\node[vertex,right=2em of v1,label=above:{$1$}] (v2) {};
						\node[vertex,below=2em of v1,label=below:{$2$}] (v3) {};
						\node[vertex,below=2em of v2,label=below:{$1$}] (v4) {};
						\draw (v1) -- (v2);
						\draw (v3) -- (v4);
						\draw (v1) -- (v4);
						\draw (v2) -- (v3);
					\end{tikzpicture}
				};
				\draw[thick] (N1) -- (N2);
			\end{tikzpicture}
		};
		\draw[rgedge,->] (G)   -- node[left,align=center]  {list coloring to\\vertex coloring} (G2);
		\draw[rgedge,->] (G)   -- node[above,align=center] {Graph of list colorings of $G$\\with short-cuts and self-loops} (K);
		\draw[rgedge,->] (G2)  -- (KGP);
		\draw[rgedge,->] (KGP) -- node[right,align=left] {graph homo-\\morphism $f$} (K);
	\end{tikzpicture}
	\end{center}
	\caption{Relations between the graphs $G$, $\Glist$, $K$ and $\KKG{p}{\Glist}$. The
			choice of $G$ and the available colors determines the other graphs
			as described in the text. The existence of the graph homomorphism
			$f$ is established by
			Lemma~\ref{lemma:gh}.\label{fig:homomorphism_example}}
\end{figure}
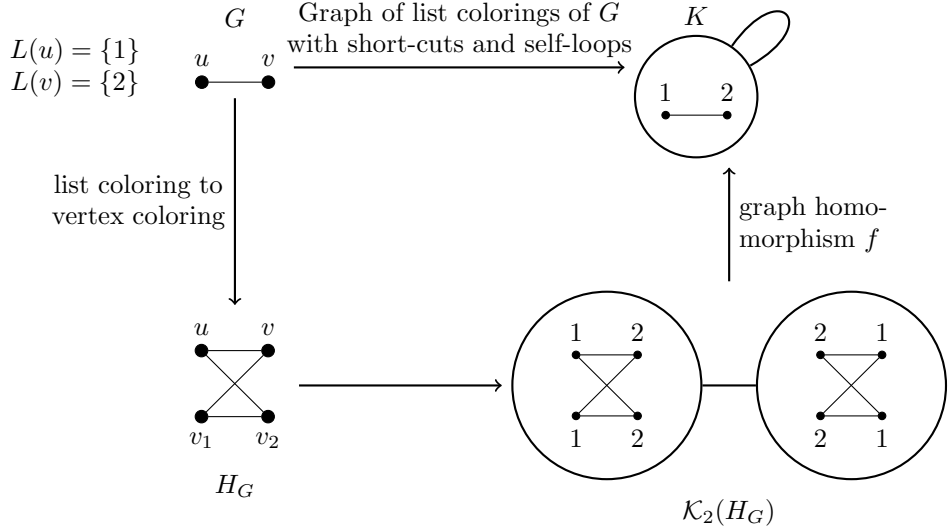

\begin{proposition}
\label{lemma:gh}
	There is a graph homomorphism $f: \KKG{p}{\Glist} \rightarrow K$.
\end{proposition}
\begin{proof}
We construct the mapping $f: V(\KKG{p}{\Glist}) \rightarrow V(K)$. Let $c \in
V(\KKG{p}{\Glist})$. First, we swap the colors $v_1,\ldots,v_p$ such that
$v_i$ has color $i$ for each $i \in \{1,\ldots,p\}$. This can be achieved by
applying the following sequence of Kempe-exchanges to the coloring $c$: For each color $j
\in \{1,\ldots,p\}$, if the current color of $v_j$ is $i \neq j$, swap the
colors in each Kempe-component of $\Glist(i,j)$. Let $c'$ be the resulting coloring.  Except for the vertices
$v_1,\ldots,v_p$ and their incident edges, $\Glist$ is just a copy of $V(G)$. Now,
pick $f(c) = \tilde c$, where $\tilde c$ is equivalent to $c'$ restricted to
the vertices $V(G) \subseteq V(\Glist)$. Clearly, $\tilde c$ is a proper
coloring of $\KKG{p}{G}$.  Due to the construction of $\Glist$, $\tilde c$
satisfies the list coloring requirements for $G$, i.e., for each $v \in V(G)$
we have $c(v) \in \alpha(v)$.  Therefore, $\tilde c \in V(K) =
V(\KKG{p}{G}[C_\pi])$.

We show that the mapping $f$ is a graph homomorphism as required. Let $c$, $d$
be colorings of $\Glist$ such that $c\uline{}d$ in $\KKG{p}{\Glist}$. Further, let \kx
be a witness of $c \keq d$. There are two cases to consider: 
\begin{enumerate}
	\item The Kempe-exchange $\kx$ does not involve any of the nodes
		$v_1,\ldots,v_p$. Then $f$ renames the color classes of the colorings
		$c$ and $d$ if required and there is a Kempe-exchange corresponding to
		$\kx$ that establishes $f(c)\uline{} f(d)$ in $K$.\label{lemma:gh:itm1}
	\item The Kempe-exchange \kx involves two nodes $u, v
		\in \{v_1,\ldots,v_p\}$. We need to consider the following two subcases. If
		$\Glist(c(u),c(v))$ is connected then $f(c) = f(d)$ and therefore,
		$f(c)\uline{}f(d)$, since each node of $K$ has a self-loop. Otherwise,
		$f(c)$ and $f(d)$ differ with respect to the color classes $c(u)$ and
		$c(v)$. We show that $f(c)$ and $f(d)$ are connected by a sequence of
		Kempe-exchanges that swaps the colors in all except a single
		Kempe-component of $\Glist(c(u),c(v))$ and thus $f(c)\uline{}f(d)$ by the
		construction of $K$ above. To obtain $f(d)$, we first apply \kx to $c$
		on $\Glist$ and then apply $f$ to the resulting coloring. The
		Kempe-exchange \kx swaps the colors of the connected component of
		$\Glist(c(u),c(v))$ containing $u$ and $v$, and then $f$ swaps the colors
		in $\Glist(c(u), c(v))$. As a result, $f(d)$ can be obtained from $f(c)$ by
		swapping the colors in $\Glist(c(u),c(v))$ except the one containing $u$
		and $v$ in the preimage $f^{-1}(V(G(c(u),c(v))))$.\label{lemma:gh:itm2} 
\end{enumerate}

In summary, for all $c,d \in V(\KKG{p}{\Glist}): c\uline{}d$ implies
$f(c)\uline{}f(d)$.
\end{proof}

The graph homomorphism $f$ induces the equivalence relation \feq on
$V(\KKG{p}{\Glist})$, that is, for $a,b \in V(\KKG{p}{\Glist}): a \eq_f b$ if $f(a) = f(b)$.

\begin{theorem}
$\KKG{p}{G}[C_\pi]$ is connected if and only if $\KKG{p}{\Glist}$ is connected.
\end{theorem}
\begin{proof}
We noted above that $\KKG{p}{G}[C_\pi]$ is connected if and only if $K$ is
connected. Let $f: \KKG{p}{\Glist} \rightarrow
K$ be the graph homomorphism from Lemma~\ref{lemma:gh}. \newline
``Only if'' part: Let $\KKG{p}{\Glist}$ be connected. Then $K$ is connected since
there is a graph homomorphism $\KKG{p}{\Glist} \rightarrow K$,
and graph homomorphisms preserve connectedness. Therefore,
$\KKG{p}{G}[C_\pi]$ is connected.\newline
``If'' part:  Let $\KKG{p}{G}[C_\pi]$ be connected. Then $K$ is
connected. Due to the first isomorphism theorem, $K \cong
\KKG{p}{\Glist}_{/\feq}$ and thus, $\KKG{p}{\Glist}_{/\feq}$ is also connected. Any two
colorings $u$, $v$ of $\Glist$ such that $u \feq v$ are connected by
Kempe-exchanges since one can be obtained from the other by permuting the
colors of the color classes.
\end{proof}

We show that using the reduction from list to vertex coloring results in a
representation of the search space which has a similar diameter compared to the
actual search space $\KKG{p}{G}[C_\pi]$.
\begin{theorem}
$\diameter(\KKG{p}{G}[C_\pi]) \leq \lfloor\frac{|V(G)|-1}{2}\rfloor\cdot\diameter(\KKG{p}{\Glist})$.
\end{theorem}
\begin{proof}
  First, note that graph homomorphisms preserve connectedness. Thus, if
  $\KKG{p}{\Glist}$ is connected so is $\KKG{p}{G}[C_\pi]$.  Now, for any adjacent
  nodes $c, d \in \KKG{p}{\Glist}$, we count how many Kempe-exchanges are required
  to get from $f(c)$ to $f(d)$ in $\KKG{p}{G}[C_\pi]$. Let \kx be the
  Kempe-exchange that is a witness of $c\uline{}d$, and let $i$ and $j$ be the
  involved color classes. If $c \feq d$ then, in the worst case, all except one
  connected component of $G(i,j)$ need to be switched to get from $c$ to $d$
  for the reasons stated in cases~\ref{lemma:gh:itm1} and~\ref{lemma:gh:itm2}
  in the proof of Lemma~\ref{lemma:gh}. There are at most $\lfloor (|V(G)|-1)/2
  \rfloor$ components and at most one Kempe-exchange is required for each of
  them. If $c \fneq d$ then there is a single Kempe-exchange on $G$ that
  establishes $f(c)\uline{}f(d)$.  Thus, a shortest path of maximum length $t$
  in $\KKG{p}{\Glist}$ corresponds to a path of length at most
  $t\cdot\lfloor(|V(G)|-1)/2\rfloor$ in $\KKG{p}{G}[C_\pi]$.
\end{proof}

\subsection{The connectedness of $\KKG{p}{\Glist}$}

Given two colorings $c$ and $c'$ of $\Glist$, the algorithm \kemperecolor
transforms $c$ into $c'$ as long as there is a sufficient number of colors
available. Due to the reduction however, $\Glist$ contains a clique on the
vertices $v_1,\ldots,v_p$, which implies that $\degeneracy(\Glist) \geq p-1$.
Therefore, according to Corollary~\ref{cor:kkg_connectedness}, the clash-free
timetables which satisfy timeslot availability requirements are connected if 
$p > \degeneracy(\Glist) \geq p-1$, that is, $\degeneracy(\Glist) = p-1$. In order to obtain less strict conditions for the
connectedness we fix the colors of the clique vertices $v_1,\ldots,v_p$ of
$\Glist$, and possibly other vertices. As a consequence, we exclude the clique
from the recoloring process, so the number of colors required by \kemperecolor is
no longer dominated by the clique.

We will first consider the general case, where the colors of some vertices $F
\subseteq V(G)$ are assumed to be fixed. We denote by $\overline{F} = V(G)
\setminus F$ be the remaining vertices. Further, let $S' \subset S(G)$ be the
vertex orderings satisfying
\begin{equation}
\label{eq:orderingcondition}
  \forall u\uline{}v \uline{}w,\; u, v \in \overline{F},\; w \in F : \; u < v  \Rightarrow w < v\enspace.
\end{equation}
That is, if $v$ is a successor of $u$ and they are adjacent, then all neighbors of
$v$ in $F$ must precede $v$. Fig.~\ref{fig:orderingcondition} shows two examples
of vertex orderings of the graph $u\uline{}v\uline{}w$. For $F = \{ w
\}$, the ordering shown in Fig.~\ref{fig:orderingcondition:positive} satisfies
the condition in Eq.~\ref{eq:orderingcondition} and the one shown in Fig.~\ref{fig:orderingcondition:negative}
does not. We will prove next that \kemperecolor does not change the color of
any vertex in $F$ if the vertices $V(G)$ are processed according to an ordering
in $S'$. In order to bound the number of colors required for our analysis, we
introduce the following generalization of the degeneracy of a graph:

\begin{definition}[Subdegeneracy]
Let $G$ be a graph and let $F \subseteq V(G)$. The \emph{subdegeneracy}
$\subdegeneracy(F, G)$ of $G$ relative to $F$ is defined as:
\begin{equation}
  \subdegeneracy(F, G) := \min_{\ordering \in S'} \max_{v \in V(G) \setminus F} \,\operatorname{pred}(v, \ordering)
  \label{eq:subdegeneracy}
\end{equation}
\end{definition}
Note that $\subdegeneracy(F, G) = \degeneracy(G)$ if $F$ is empty and
$\subdegeneracy(F, G) \leq \degeneracy(G)$ otherwise. Intuitively, we are
looking for a vertex ordering in $\orderings'$ that minimizes the maximum
number of adjacent predecessors of any vertex, however, the number of
predecessors of any vertex in $F$ is irrelevant.

\begin{remark}
  The ordering constraints in Eq.~\eqref{eq:orderingcondition} are reminiscent
  of the NP-complete problem~\cite[MS1 and MS2]{GJ:79}. However, an ordering $\ordering \in \orderings'$
  can be found in polynomial time (if one exists) by a reduction to
  $\textsc{2SAT}$: The reduction adds for each implication in
  Eq.~\eqref{eq:orderingcondition} an appropriate $\textsc{2SAT}$ clause\footnote{We would like to thank Alexander Ra\ss\ for this observation.}.
  However, the complexity of determining the subdegeneracy, i.e., the value of
  the min-max expression in Eq.~\eqref{eq:subdegeneracy}, is an open problem.
\end{remark}

\begin{figure}
	\centering\subfloat[Ordering compatible with Eq.~\eqref{eq:orderingcondition}\label{fig:orderingcondition:positive}]{
		\begin{tikzpicture}
			\node[] (U)	{$u$};
			\node[right of=U]	(comp1) {${}<{}$};
			\node[right of=comp1] (W)	{$w$};
			\node[right of=W]	(comp2) {${}<{}$};
			\node[right of=comp2] (V)	{$v$}
				edge [bend right=55] (W)
				edge [bend right=55] (U);
		\end{tikzpicture}
	}\qquad
	\subfloat[Ordering incompatible with Eq.~\eqref{eq:orderingcondition}\label{fig:orderingcondition:negative}]{
		\begin{tikzpicture}
			\node[] (U)	{$u$};
			\node[right of=U]	(comp1) {${}<{}$};
			\node[right of=comp1] (V)	{$v$}
				edge [bend right=55] (U);
			\node[right of=V]	(comp2) {${}<{}$};
			\node[right of=comp2] (W)	{$w$}
				edge [bend right=55] (V);
		\end{tikzpicture}
	}
	\caption{Two vertex orderings of the graph $u\protect\uline{}v\protect\uline{}w$, $F = \{w\}$.\label{fig:orderingcondition}}
\end{figure}
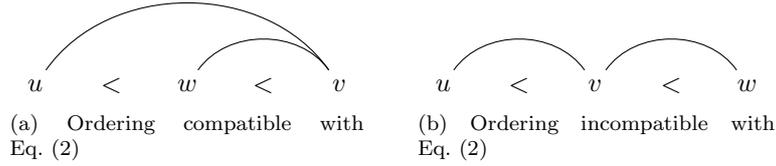

\begin{theorem}
\label{thm:subdegeneracy_connectedness}
Let $c$, $c'$ be $k$-colorings of $G$ that agree on $F$. Then
\kemperecolor returns a sequence of Kempe-exchanges such that 
\begin{enumerate}
	\item all intermediate colorings also agree on $F$, and
	\item no more than $\subdegeneracy(F, G)+1$ colors are required.
\end{enumerate}
\end{theorem}
\begin{proof}
We first show that the colors of the vertices $F$ are not changed by \kemperecolor.
Assume for a contradiction that in some intermediate coloring a vertex $w \in
F$ has a color different from $c(w)$. Then $w$ has been recolored because a
neighbor $u$ of $w$ preceding it in $\ordering$ received color $c(w)$. There
are two possible reasons: Either $u$ was recolored to $c(w)$ because $c'(u) =
c(w)$, but then $c'(w) \neq c(w)$, a contradiction. If this is not the case,
then $u$ was recolored in case~\ref{alg:kemperecolor:one}
or~\ref{alg:kemperecolor:two} of \kemperecolor, because of a neighbor $v$
preceding it. But this is a contradiction to $\ordering \in S'$.

We now show that $\subdegeneracy(F, G)+1$ colors are sufficient. Since the
vertices in $F$ are never recolored, we consider only the vertices
$\overline{F}$. An unused color may be picked for a vertex $v \in \overline{F}$
in case~\ref{alg:kemperecolor:two} of Algorithm~\ref{alg:kemperecolor}. For
each $v \in \overline{F}$, there are at most $\subdegeneracy(F, G)$ neighbors
of $v$ preceding it, and there are at most $\subdegeneracy(F, G)-1$ colors
different from the color of $v$ present among these vertices. Thus, there is at
least one other color available for $v$.
\end{proof}

Recall that for any two clash-free timetables we can assume the colors of the
clique vertices $v_1,\ldots,v_p$ of $\Glist$ to be fixed. If we pick $F \subseteq
V(\Glist)$ such that any two colorings of $\Glist$ which agree on the clique also agree
on $F$, then we obtain the following:

\begin{corollary}
\label{cor:subdegeneracy_connectedness}
The clash-free timetables that satisfy timeslot availability requirements are
connected if $|\periods| > \subdegeneracy(F, \Glist)$.
\end{corollary}

We return to the general setting and propose a heuristic approach to finding a
witness vertex ordering of $\subdegeneracy(F, G)$ for any graph $G$ and $F
\subseteq V(G)$. Let $\tilde S \subseteq S'$ be the vertex orderings such that
the vertices $F$ precede all other vertices.  Recall that for any graph $G$ a
witness vertex ordering of the degeneracy $\degeneracy(G)$ can be found by
repeatedly removing vertices of minimal degree. In a similar fashion, we can
determine the value
\[
	\lambda(F, G) := \min_{\ordering \in \tilde \orderings} \max_{v \in V(G) \setminus F} \operatorname{pred}(v, \ordering)\enspace.
\]
Moreover, this value is equivalently characterized by a max-min expression,
analogous to the two characterizations of the degeneracy shown in
Eq.~\eqref{eq:degeneracy}:

\begin{algorithm}
\caption{\vertexelimination}
\label{alg:vertexelimination}

\DontPrintSemicolon
\SetKwInOut{Input}{input}
\SetKwInOut{InOut}{in/out}
\SetKwInOut{Output}{output}
\SetKwInOut{Data}{data}

\SetKwData{GG}{$G$}
\SetKwData{CC}{$F$}
\SetKwData{DD}{$\overline{F}$}
\SetKwData{HH}{$H$}
\SetKw{KwTo}{downto}%

\Input{graph \GG, vertices $\CC \subseteq V(G)$}
\Output{ordering $v_1,\ldots,v_{|\DD|}$ of the vertices $\DD = V(G) \setminus \CC$}
\BlankLine
$\GG_{|D|} \longleftarrow \GG$ \;
\For{$i \longleftarrow |\DD|$ \KwTo $1$}
{
	choose $v_i$ from $\argmin_{v \in \DD} \{ \vdegree(v, \GG_i) \}$ \;
	$\GG_{i-1} \longleftarrow \GG_i - v_i$.
}
\Return $v_1,\ldots,v_{|\DD|}$\;
\end{algorithm}

\begin{theorem}
For any graph $G$ and $F \subseteq V(G)$,
\[
	\lambda(F, G) = \min_{\ordering \in \tilde \orderings} \max_{v \in V(G) \setminus F} \operatorname{pred}(v, \ordering) = \max_{G[F] \subseteq H \subseteq G} \min_{v \in V(H)\setminus F} \{d_{H}(v)\}\enspace.
\]
Furthermore, \vertexelimination produces a witness vertex ordering of $\lambda(F, G)$.
\end{theorem}
\begin{proof}
The proof is based on the remark on the optimality of \vertexelimination
in~\cite{Matula:68}. Let $\ell = |\overline{F}|$ and for an ordering
$v_1,\ldots,v_\ell$ of $\overline{F}$ let $G_i = G[F \cup \{v_1,\ldots,v_\ell\}]$.
Further, let
\[
	\hat \mindegree := \displaystyle\max_{G[F] \subseteq H \subseteq G} \displaystyle\min_{v \in V(H) \setminus F} \{ \vdegree(v, H) \} \enspace.
\]
Intuitively, $\hat \mindegree$ is analogous to the degeneracy of $G$, but the
vertices $F$ are irrelevant.  If an ordering $\ordering = v_1,\ldots,v_\ell$ of
$\overline{F}$ is an output of
\vertexelimination then
\begin{align*}
	\max_{1\leq i\leq \ell} \operatorname{pred}(v_i, \ordering) & = \max_{1 \leq i \leq \ell}\{\vdegree(v_i, G_i)\} \\
			& = \displaystyle\max_{1 \leq i \leq \ell} \displaystyle\min_{v \in V(G_i)\setminus F} \{ \vdegree(v, G_i) \} \leq \hat \mindegree\enspace.
\end{align*}
The graphs $G_i$ coincide with those in Algorithm~\ref{alg:vertexelimination}. 

Now let $H^{*}$ be a graph such that
$G[F] \subseteq H^{*} \subseteq G$ and
\[
	\min_{v \in V(H^{*}) \setminus F} \{ \vdegree (v, H) \} = \hat \mindegree \enspace.
\]
Let $v_1,\ldots,v_\ell$ be any ordering of $\overline{F}$ and let $i$ be the smallest
index such that $H^{*} \subseteq G_i$. Then $v_i$ must be a vertex of $H^{*}$
and $\vdegree(v_i, G_i) \geq \hat \mindegree$.
Therefore, for any ordering $v_1,\ldots,v_\ell$ of $\overline{F}$,
 $\max_{1 \leq j \leq \ell} \{ \vdegree(v_j, G_j) \} \geq \hat \mindegree$, with
equality if the vertex ordering is an output of \vertexelimination.
\end{proof}

Certainly, the optimality of \vertexelimination is only established with
respect to the subset $\tilde S \subseteq S'$.  The vertex ordering obtained
from the algorithm can potentially be improved by the following post-processing
step: Let $v_1,\ldots,v_{|\overline{F}|}$ be an output of \vertexelimination
and let $k$ be the largest number such that $v_1,\ldots,v_k$ are independent.
Then the vertices $v_1,\ldots,v_k$ can be moved before the vertices $F$ in the
ordering without violating condition~\eqref{eq:orderingcondition}.  The
resulting ordering $\ordering' \in
S'$ is not in $\tilde S$ and can thus not be generated by \vertexelimination.
There is a potential advantage because the construction guarantees that
$\max_{v \in \overline{F}} \operatorname{pred} (v, \ordering') \leq \max_{v \in
	\overline{F}} \operatorname{pred}(v, \ordering)$.  

In summary, the heuristic for computing a vertex ordering $\ordering \in \orderings'(G)$
such that the value $\max_{v \in \overline{F}} \operatorname{pred}(v, \ordering)$ is close to
$\subdegeneracy(F, G)$ performs the following two steps:
\begin{enumerate}
	\item Run \vertexelimination to generate an ordering $v_1,\ldots,v_{|\overline{F}|}$ of the vertices $\overline{F}$.
	\item Let $k \in \mathbb{N}$ be the largest number such that
		$v_1,\ldots,v_k$ are independent in $G$. Move the vertices
		$v_1,\ldots,v_k$ before the vertices $F$ in the ordering. 
\end{enumerate}

\section{Results}
\label{sec:results}

We use the theoretical insights from the previous section to establish the
connectedness of clash-free timetables for a range of \ac{UTP} benchmark
instances. Given a conflict graph $G$, by
Corollary~\ref{cor:kkg_connectedness}, the reconfiguration graph of clash-free
timetables is connected if $p > \degeneracy(G)$. If timeslot availability
requirements are present, we first use the reduction from list to graph
coloring described in Section~\ref{sec:background:vcol} to construct the graph
$\Glist$, which contains the additional clique
$v_1,\ldots,v_p$. We then use the heuristic from the previous section to determine a
bound $\subdegeneracyU(F, \Glist) \geq \subdegeneracy(F, \Glist)$. The set $F$ of
vertices with ``fixed'' colors contains the clique vertices $v_1,\ldots,v_p$
and any other node of $G$ which has only a single available timeslot/color:
\begin{equation}
	F =  \{ v \in V(\Glist) \mid |\Gamma(v) \cap \{v_1,\ldots,v_p\}| = p-1 \}\enspace,\label{eq:fixedvertices}
\end{equation}
where $\Gamma(v)$ denotes the set of vertices adjacent to $v$.
By Corollary~\ref{cor:subdegeneracy_connectedness}, the reconfiguration graphs
of the clash-free timetables that satisfy timeslot availability requirements are
connected if $p > \subdegeneracy(F, \Glist)$.

\begin{table}
	\begin{center}
	\caption{For each instance from the \ac{CB-CTT}, \ac{PE-CTT}, and Erlangen
		instance sets, we give the number $p$ of timeslots, $\degeneracy(G)$ and
			an upper bound $\subdegeneracyU(F, \Glist) \geq \subdegeneracy(F, \Glist)$
			produced by the heuristic. Values in bold face indicate the
			connectedness of the clash-free timetables according to
			Corollaries~\ref{cor:kkg_connectedness}
			and~\ref{cor:subdegeneracy_connectedness}.  \label{tab:connectedness}}
	\small\begin{tabular}{@{}l|c|c|c||l|c|c|c}
	\toprule
	instance & $p$ & $\degeneracy(G)$ &    $\subdegeneracyU(F, \Glist)$ & instance & $p$ & $\degeneracy(G)$ & $\subdegeneracyU(F, \Glist)$\tabularnewline\midrule
	{\tt comp01} & 30 & \bf 23 &   \bf 24      & {\tt ITC2\_i01} &  45  &  91 &   109    \tabularnewline
	{\tt comp02} & 25 & \bf 23 &   30   	   & {\tt ITC2\_i02} &  45  &  99 &   119    \tabularnewline
	{\tt comp03} & 25 & \bf 22 &   27    	   & {\tt ITC2\_i03} &  45  &  73 &   92     \tabularnewline
	{\tt comp04} & 25 & \bf 17 &   25     	   & {\tt ITC2\_i04} &  45  &  78 &   100    \tabularnewline 
	{\tt comp05} & 36 & \bf 26 &   43     	   & {\tt ITC2\_i05} &  45  &  81 &   99     \tabularnewline
	{\tt comp06} & 25 & \bf 17 &   28     	   & {\tt ITC2\_i06} &  45  &  80 &   100    \tabularnewline
	{\tt comp07} & 25 & \bf 20 &   \bf 24      & {\tt ITC2\_i07} &  45  &  80 &   106    \tabularnewline
	{\tt comp08} & 25 & \bf 20 &   \bf 24      & {\tt ITC2\_i08} &  45  &  69 &   97     \tabularnewline
	{\tt comp09} & 25 & \bf 22 &   25     	   & {\tt ITC2\_i09} &  45  &  89 &   108    \tabularnewline
	{\tt comp10} & 25 & \bf 18 &   27     	   & {\tt ITC2\_i10} &  45  &  97 &   116    \tabularnewline
	{\tt comp11} & 45 & \bf 27 &   \bf 27      & {\tt ITC2\_i11} &  45  &  75 &   93     \tabularnewline
	{\tt comp12} & 36 & \bf 22 &   40   	   & {\tt ITC2\_i12} &  45  &  91 &   109    \tabularnewline
	{\tt comp13} & 25 & \bf 17 &   \bf 22	   & {\tt ITC2\_i13} &  45  &  87 &   106    \tabularnewline 
	{\tt comp14} & 25 & \bf 17 &   \bf 23	   & {\tt ITC2\_i14} &  45  &  87 &   107    \tabularnewline
	{\tt comp15} & 25 & \bf 22 &   27    	   & {\tt ITC2\_i15} &  45  &  79 &   106    \tabularnewline
	{\tt comp16} & 25 & \bf 18 &   25    	   & {\tt ITC2\_i16} &  45  &  55 &   83     \tabularnewline
	{\tt comp17} & 25 & \bf 17 &   25    	   & {\tt ITC2\_i17} &  45  &  50 &   71     \tabularnewline
	{\tt comp18} & 36 & \bf 14 &   \bf 32	   & {\tt ITC2\_i18} &  45  &  91 &   112    \tabularnewline
	{\tt comp19} & 25 & \bf 23 &   27    	   & {\tt ITC2\_i19} &  45  &  101&   120    \tabularnewline
	{\tt comp20} & 25 & \bf 19 &   \bf 23	   & {\tt ITC2\_i20} &  45  &  73 &   92     \tabularnewline
	{\tt comp21} & 25 & \bf 23 &   28    	   & {\tt ITC2\_i21} &  45  &  72 &   90     \tabularnewline
	{\tt erl.2011-2} & 30 &  \bf 22 & 32   	   & {\tt ITC2\_i22} &  45  &  98 &   118    \tabularnewline
	{\tt erl.2012-1} & 30 &  \bf 14 & 31   	   & {\tt ITC2\_i23} &  45  &  117&   128    \tabularnewline
	{\tt erl.2012-2} & 30 &  \bf 20 & 32   	   & {\tt ITC2\_i24} &  45  &  77 &   97      \tabularnewline
	{\tt erl.2013-1} & 30 &  \bf 16 & 30   	   & {\tt toy}    & 20 & \bf 10 &  \bf 11  \tabularnewline
	\bottomrule
	\end{tabular}                       
	\end{center}
\end{table}

Table~\ref{tab:connectedness} indicates the connectedness of the clash-free
timetables according to Corollaries~\ref{cor:kkg_connectedness}
and~\ref{cor:subdegeneracy_connectedness} for instances from the \ac{CB-CTT},
\ac{PE-CTT} benchmark sets, as well as instances from the University of
Erlangen-N{\"u}rnberg. All instances are available from the website \cite{ITC2007:web}. The instances
\texttt{comp01},\ldots,\texttt{comp21} are from the \ac{CB-CTT} track of the
\ac{ITC2007} competition. The instances
\texttt{ITC2\_i01},\ldots,\texttt{ITC2\_i24} are from the \ac{PE-CTT} track of
the same competition. The \erlangen instances are large real-world instances
from the engineering department of the University of Erlangen-N{\"u}rnberg. The
\texttt{toy} instance is a small example instance from the
website \cite{ITC2007:web}.  For each instance we give the
number of timeslots $p$, the degeneracy of the conflict graph $\degeneracy(G)$,
and the bound $\subdegeneracyU(F, \Glist) \geq \subdegeneracy(F, \Glist)$. Table
entries in bold face indicate that the corresponding value $\degeneracy(G)$ or
$\subdegeneracyU(F, \Glist)$ certifies the connectedness of the clash-free
timetables.

According to the data in Table~\ref{tab:connectedness} the clash-free timetables for all \ac{CB-CTT}
and \erlangen instances are connected, while the conditions imposed by 
Corollary~\ref{cor:kkg_connectedness} are not satisfied for any of the
\ac{PE-CTT} instances. For eight \ac{CB-CTT} instances, the upper bound on
$\subdegeneracy(F, \Glist)$ is sufficient to show that the reconfiguration graphs
are connected in the presence of timeslot availability constraints. For the
\ac{PE-CTT} instances, since neither $\degeneracy(G)$ nor $\subdegeneracyU(F,
\Glist)$ certifies the connectedness of the reconfiguration graphs, better
bounds on $\subdegeneracy(F, \Glist)$ are of no use, since $\subdegeneracy(F,
\Glist) \geq \degeneracy(G)$. Therefore, new techniques are needed for proving
the connectedness (or disconnectedness) of the reconfiguration graphs for these
instances. A possible reason for this structural difference between the
\ac{CB-CTT} and \ac{PE-CTT} instances is that the course specification in the
former leads to lots of small cliques in the conflict graph, a fact that we
will use shortly to determine the subdegeneracy of the conflict graph for the
\ac{CB-CTT} instance \texttt{toy}. In contrast, in the \ac{PE-CTT} problem
formulation and also the instances from~\cite{LP:instances}, the event conflicts
depend on the students' individual choices, which apparently leads to denser
graphs, i.e., graphs with higher degeneracy.

\begin{table}
	\begin{center}
	\caption{The connectedness of the clash-free timetables for the
		instances from~\cite{LP:instances}. For each instance we
		give the degeneracy $\degeneracy(G)$ of the conflict graph $G$.  Values
		in bold face indicate the connectedness of clash-free timetables according to Corollary~\ref{cor:kkg_connectedness}.
		\label{tab:metaheuristicdegeneracy}}
	\small\begin{tabular}{@{}l|c||l|c||l|c}
	\toprule
	instance &  $\degeneracy(G)$ &    instance &  $\degeneracy(G)$ & instance &  $\degeneracy(G)$ \tabularnewline\midrule
\texttt{small\_1}    &  54   & \texttt{med\_1}      &  59    &   \texttt{big\_1}    &  60   \tabularnewline        
\texttt{small\_2}    &\bf 41 & \texttt{med\_2}      &  67    &   \texttt{big\_2}    &  68   \tabularnewline        
\texttt{small\_3}    &  98   & \texttt{med\_3}      &  67    &   \texttt{big\_3}    &  64   \tabularnewline        
\texttt{small\_4}    &  69   & \texttt{med\_4}      &  69    &   \texttt{big\_4}    &  80   \tabularnewline        
\texttt{small\_5}    &  84   & \texttt{med\_5}      &  87    &   \texttt{big\_5}    &  75   \tabularnewline        
\texttt{small\_6}    &\bf 24 & \texttt{med\_6}      &  101   &   \texttt{big\_6}    &  93   \tabularnewline        
\texttt{small\_7}    &  68   & \texttt{med\_7}      &  120   &   \texttt{big\_7}    &  111  \tabularnewline        
\texttt{small\_8}    &  84   & \texttt{med\_8}      &  98    &   \texttt{big\_8}    &  82   \tabularnewline        
\texttt{small\_9}    &  124  & \texttt{med\_9}      &  121   &   \texttt{big\_9}    &  77   \tabularnewline        
\texttt{small\_10}   &  136  & \texttt{med\_10}     &  64    &   \texttt{big\_10}   &  77   \tabularnewline        
\texttt{small\_11}   &\bf 34 & \texttt{med\_11}     &  97    &   \texttt{big\_11}   &  76   \tabularnewline        
\texttt{small\_12}   &\bf 22 & \texttt{med\_12}     &  78    &   \texttt{big\_12}   &  76   \tabularnewline        
\texttt{small\_13}   &  146  & \texttt{med\_13}     &  105   &   \texttt{big\_13}   &  84   \tabularnewline        
\texttt{small\_14}   &  100  & \texttt{med\_14}     &  92    &   \texttt{big\_14}   &  74   \tabularnewline        
\texttt{small\_15}   &  79   & \texttt{med\_15}     &  101   &   \texttt{big\_15}   &  127  \tabularnewline        
\texttt{small\_16}   &  118  & \texttt{med\_16}     &  145   &   \texttt{big\_16}   &  115  \tabularnewline        
\texttt{small\_17}   &  120  & \texttt{med\_17}     &  126   &   \texttt{big\_17}   &  184  \tabularnewline        
\texttt{small\_18}   &  60   & \texttt{med\_18}     &  188   &   \texttt{big\_18}   &  131  \tabularnewline        
\texttt{small\_19}   &  141  & \texttt{med\_19}     &  173   &   \texttt{big\_19}   &  159  \tabularnewline        
\texttt{small\_20}   &\bf 28 & \texttt{med\_20}     &  153   &   \texttt{big\_20}   &  144  \tabularnewline        
	\bottomrule
	\end{tabular}                       
	\end{center}
\end{table}

\begin{table}
	\begin{center}
	\caption{The connectedness of the clash-free timetables for the
		Metaheuristic Network instances from~\cite{MN:instances}.  For each instance we
			give the degeneracy $\degeneracy(G)$ of the conflict graph $G$.
			Values in bold face indicate the connectedness of clash-free
			timetables according to Corollary~\ref{cor:kkg_connectedness}.
			\label{tab:metaheuristicdegeneracy2}}
	\small\begin{tabular}{@{}l|c||l|c||l|c}
	\toprule
	instance &  $\degeneracy(G)$ &    instance &  $\degeneracy(G)$ & instance &  $\degeneracy(G)$ \tabularnewline
	\midrule
	\texttt{easy01}  &  \bf 15 & \texttt{medium01} &   49 &\texttt{hard01}  &   68 \tabularnewline
	\texttt{easy02}  &  \bf 19 & \texttt{medium02} &   53 &\texttt{hard02}  &   67 \tabularnewline
	\texttt{easy03}  &  \bf 13 & \texttt{medium03} &   52 &&\tabularnewline
	\texttt{easy04}  &  \bf 12 & \texttt{medium04} &   51 &&\tabularnewline
	\texttt{easy05}  &  \bf 20 & \texttt{medium05} &   47 &&\tabularnewline
	\bottomrule
	\end{tabular}                       
	\end{center}
\end{table}

In Tables~\ref{tab:metaheuristicdegeneracy}
and~\ref{tab:metaheuristicdegeneracy2}, the degeneracy values of the
corresponding conflict graphs are given for the instance sets
from \cite{LP:instances} and \cite{MN:instances}. On these
instances, each timeslot is available for each event. Values in bold face
indicate the connectedness of clash-free timetables is established by
Corollary~\ref{cor:kkg_connectedness}.

Finally, we will show that for the \ac{CB-CTT} instance \texttt{toy}, the proposed
heuristic yields an optimal vertex ordering, i.e., a witness for
$\subdegeneracy(F, \Glist)$. The instance has in total 20 timeslots and 16 events.
In the \ac{CB-CTT} formulation, the events are grouped into courses. Any two
events of a course are conflicting, that is, the events of a course are a
clique in the conflict graph. Whenever
two courses are conflicting, no two of the corresponding events may be
scheduled in the same timeslot. For completeness, the relevant data on events,
conflicts and unavailable timeslots is given in Table~\ref{tab:toy}.

\begin{table}
  \caption{Instance data of the instance \texttt{toy}, available from the website~\cite{ITC2007:web}.}
  \begin{center}
	\begin{tabular}{r|c|c|c}
	  \toprule
	  Course	&	Events	&	Conflicts 	&	Unavailable timeslots	\tabularnewline\midrule
	  TecCos	&	5		&	SceCosC, ArcTec, Geotec		&	8,9,14,15				\tabularnewline
	  ArcTec	&	3		&	SceCosC, TecCos				&	16,17,18,19					\tabularnewline
	  SceCosC	&	3		&	ArcTec, TecCos				&	--\tabularnewline
	  Geotec	&	5		&	TecCos						&	--\tabularnewline
	  \bottomrule
	\end{tabular}
  \end{center}
  \label{tab:toy}
\end{table}

Let $G$ be the conflict graph of the instance \texttt{toy}
and let $\Glist$ be the graph that results from the reduction from list to
graph coloring. If two courses are in conflict, then the events of both
courses are a clique in $G$. 
If certain timeslots are unavailable for a
course, then the events of the course and then these timeslots form a clique in
$\Glist$. Figure~\ref{fig:vertexordering} shows
a succinct representation of
the graph $\Glist$. The nodes $T$, $A$, $S$ and $G$ correspond to event cliques
of the courses TecCos, ArcTec, SceCosC, and Geotec, respectively. 
The node $P_1$ represents the timeslots marked unavailable for the course
ArcTec and the node $P_2$ represents the timeslots unavailable for SceCosC.
Since no other timeslots are excluded, the
corresponding vertices in the graph $\Glist$ will not contribute to the
subdegeneracy and can be ignored.  As a result we get a clique on eight nodes
that model the timeslot availability requirements. This clique is divided
separated into the two cliques $P_1$ and $P_2$. Two nodes of the shown graph
are connected whenever all nodes of the two corresponding cliques are
connected.

\begin{figure}
	\begin{center}
	\begin{tikzpicture}[every pin edge/.style={black!50},vertex/.style={shape=circle,draw,minimum size=2.5em,node distance=6em},conflict/.style={draw,thick}]
		\node[vertex] (p1) [pin=below:{$K_4$}] {$P_1$};
		\node[vertex] (p2) [right of=p1,pin=below:{$K_4$}] {$P_2$};
		\node[vertex] (T)  [right of=p2,pin=below:{$K_5$}] {T};
		\node[vertex] (A)  [right of=T,pin=below:{$K_3$}] {A};
		\node[vertex] (S)  [right of=A,pin=below:{$K_3$}] {S};
		\node[vertex] (G)  [right of=S,pin=below:{$K_5$}] {G};

		\draw[conflict] (p1) to [bend left=55] (p2);
		\draw[conflict] (p1) to [bend left=55] (A);
		\draw[conflict] (p2) to [bend left=55] (T);
		\draw[conflict] (T) to [bend left=55] (A);
		\draw[conflict] (T) to [bend left=55] (S);
		\draw[conflict] (A) to [bend left=55] (S);
		\draw[conflict] (T) to [bend left=55] (G);
	\end{tikzpicture}
	\end{center}
	\caption{Succinct representation of the graph $\Glist$, where $G$ is the
	  conflict graph of the instance \texttt{toy}. All nodes represent cliques
	  as denoted indicated the nodes. Arranging the clique vertices in the
	  shown left-to-right ordering yields a witness of $\subdegeneracy(F,
	  \Glist) = 11$.\label{fig:vertexordering}}
\end{figure}
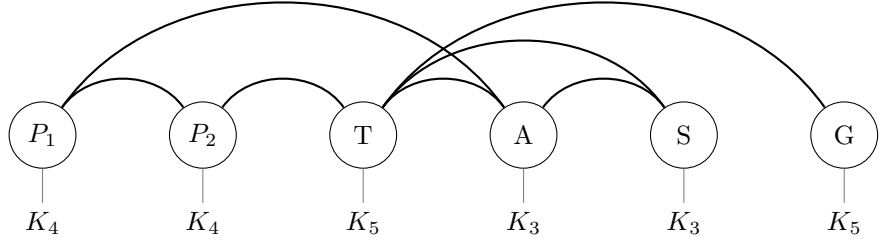
Let $F = V(P_1) \cup V(P_2)$
and let $\ordering \in \orderings(\Glist)$ such that the cliques are
arranged in the order $P_1,P_2,T,A,S,G$ with some arbitrary choice of the
relative ordering of the vertices within each clique. This ordering is a
possible output of the algorithm \vertexelimination. From
\[
	\max_{v \in V(\Glist) \setminus F} \operatorname{pred}(v,\ordering) = 11\enspace,
\]
we can conclude that $\subdegeneracy(F, \Glist) \leq 11$. 

\begin{proposition}
For the instance \texttt{toy}, $\subdegeneracy(F, \Glist) = 11$.
\end{proposition}
\begin{proof}
Let $\ordering$ be an ordering of $V(\Glist)$ and $V' \subseteq V(\Glist) \setminus F$. The
maximum number of predecessors adjacent to any vertex of $V'$ in $\Glist$ is denoted by
\[
	p(V',\ordering) = \max_{v \in V'} \operatorname{pred}(v,\ordering)\enspace.
\]	
Note that for a clique $K \in \{ T, A, S, G \}$, the value $p(K,\ordering)$ is
determined by the last vertex of $K$ in $\ordering$. Thus, the value of
$p(K,\ordering)$ depends only on the relative order of the last vertices of the
cliques $\{T, A, S, G\}$ in $\ordering$. Let $\tilde S$ be the vertex orderings
of $\Glist$ such the vertices $F$ precede all other vertices of $\Glist$ and let $\hat
S$ be the total orderings of $\{T, A, S, G\}$. For each ordering $\ordering' \in \hat
S$ we can pick an ordering $\ell(\ordering')$ of $\Glist$ that is compatible with
$\ordering'$ in the sense that the relative ordering of the last vertices of
the cliques is in accordance with $\ordering'$. We have, 
\[
\subdegeneracy(F, \Glist) =
\min_{\ordering \in \tilde S} \max_{K \in \{T, A, S, G\} } p(K, \sigma) = \min_{\ordering' \in \hat S} \max_{K \in \{T, A, S, G\}} p(K, \ell(\sigma'))\enspace.
\]
For any ordering $\ordering' \in \hat S$ such that $G < T$, we have $p(T,
\ell(\ordering')) \geq 13$, because the last vertex of $T$ has
at least 13 adjacent predecessors in $\Glist$. Thus, we only need to consider
orderings such that $G > T$. Furthermore, since no vertex of $G$ is adjacent to
any vertex of $A$ or $S$, changing the relative order of $A$ and $G$ or $S$ and
$G$ does not change the number of adjacent predecessors. Hence, we can assume
$G$ is a maximum in any ordering of interest. We enumerate the values of
$p(K,\ell(\ordering'))$ all for $K \in \{T, A, S, G\}$ for the 6 permutations
of $\{T, A, S\}$:
\begin{center}
\begin{tabular}{l|c|c|c|c}
\toprule
clique ordering $\ordering' \in \hat S$     & $p(T,\ell(\ordering'))$ & $p(A,\ell(\ordering'))$ & $p(S,\ell(\ordering'))$ & $p(G,\ell(\ordering'))$ \tabularnewline\midrule
$T, A, S, G$ & 8      & 11     & 10     & 9      \tabularnewline
$T, S, A, G$ & 8      & 14     & 7      & 9      \tabularnewline
$A, T, S, G$ & 11     & 6      & 10     & 9      \tabularnewline
$S, T, A, G$ & 11     & 11     & 2      & 9      \tabularnewline
$A, S, T, G$ & 14     & 6      & 5      & 9      \tabularnewline
$S, A, T, G$ & 14     & 9      & 2      & 9      \tabularnewline
\bottomrule
\end{tabular}
\end{center}
Thus, 
\[
	\subdegeneracy(F, \Glist) = \min_{\ordering' \in \hat S} \max_{K \in \{T, A, S, G\}} p(K, \ell(\ordering')) = 11
\]
\end{proof}

We can conclude that the proposed heuristic produces a witness of
$\subdegeneracy(F, \Glist) = 11$ on the instance \texttt{toy}.

\section{Conclusions}

We investigated the connectedness of clash-free timetables with respect to the
Kempe-exchange operation. This investigation is related to the connectedness of
the search space of timetabling problem instances, which is a desirable
property, for example for two-step algorithms using the Kempe-exchange during
the optimization step. We include timeslot availability requirements in our
analysis and derive improved conditions for the connectedness of clash-free
timetables in this setting. For this purpose, we introduced the notion of
subdegeneracy, which generalizes the degeneracy of a graph. The complexity of
determining the subdegeneracy is an interesting open problem.
We further showed that our representation of the 
search space of clash-free timetables that satisfy timeslot availability
requirements is a suitable one with respect to the connectedness properties and
the diameter of the search space. Our results indicate the connectedness of the
clash-free timetables for a number of benchmark instances.

For future research, other properties of feasible timetables such as
overlap-freeness may be considered as well. Furthermore, two kinds of possible
improvements may be considered with respect to establishing the connectedness
of clash-free timetables in the presence of timeslot availability requirements:
Both, a better analysis of Algorithm~\ref{alg:kemperecolor} and a better
heuristic approach (or exact algorithm) for determining the subdegeneracy may
lead to a lower number of timeslots required to certify the connectedness of
clash-free timetables.

\bibliographystyle{plain}
\bibliography{literature}   

\end{document}